\documentclass[letterpaper,10 pt, conference]{ieeeconf}

\IEEEoverridecommandlockouts                              
\usepackage{amsmath,amssymb,amsfonts,amsthm} 
\usepackage{mathtools,bbm}
\usepackage{tikz, epic,eepic}
\usepackage{caption}
\usepackage{subcaption}
\usetikzlibrary{shapes,arrows}
\usetikzlibrary{automata}
\usepackage{bbm}
\usepackage{bm}
\usepackage{bbold}
\usepackage{epsfig, amsbsy}
\usepackage{mathtools}
\usepackage{graphicx}
\usepackage{caption}
\usepackage{subcaption}
\usepackage{algorithm}
\usepackage[noend]{algpseudocode}
\usepackage{epstopdf}
\usepackage{multirow,array}
\usepackage{cite}

\let\labelindent\relax
\usepackage{enumitem}
\theoremstyle{remark}

\newtheorem{theorem}{Theorem}
\newtheorem{proposition}{Proposition}

\newtheorem{lemma}{Lemma}
\newtheorem{definition}{Definition}
\newtheorem{example}{Example}
{\itshape}{\rmfamily}

\newtheorem{remark}{Remark}

\def\Bcdf{F_B}
\def\Ecdf{F_E}

\title{\LARGE \bf
The Division of Assets in Multiagent Systems: \\ A Case Study in Team Blotto Games   
}

\author{Keith Paarporn, Rahul Chandan, Mahnoosh Alizadeh and Jason R. Marden 
\thanks{K. Paarporn ({\tt\small kpaarporn@ucsb.edu}), R. Chandan ({\tt\small rchandan@ucsb.edu}), M. Alizadeh ({\tt\small alizadeh@ucsb.edu}), and J. R. Marden ({\tt\small jrmarden@ece.ucsb.edu}) are with the Department of Electrical and Computer Engineering at the University of California, Santa Barbara, CA.} \thanks{This work is supported by UCOP Grant LFR-18-548175, ONR grant \#N00014-20-1-2359, and AFOSR grant \#FA9550-20-1-0054}
}

\begin{document}

\maketitle
\thispagestyle{empty}
\pagestyle{empty}

\begin{abstract}
    Multi-agent systems are designed to concurrently accomplish a diverse set of tasks at unprecedented scale. Here, the central problems faced by a system operator are to decide (i) how to divide available resources amongst the agents assigned to tasks and (ii) how to coordinate the behavior of the agents to optimize the efficiency of the resulting collective behavior. The focus of this paper is on problem (i), where we seek to characterize the impact of the division of resources on the best-case efficiency of the resulting collective behavior.  Specifically, we focus on a team Colonel Blotto game where there are two sub-colonels competing against a common adversary in a two battlefield environment. Here, each sub-colonel is assigned a given resource budget and is required to allocate these resources independent of the other sub-colonel. However, their success is dependent on the allocation strategy of both sub-colonels.  The central focus of this manuscript is on how to divide a common pool of resources among the two sub-colonels to optimize the resulting best-case efficiency guarantees.  Intuitively, one would imagine that the more balanced the division of resources, the worse the performance, as such divisions restrict the sub-colonels' ability to employ joint randomized strategies that tend to be necessary for optimizing performance guarantees.  However, the main result of this paper demonstrates that this intuition is actually incorrect. A more balanced division of resources can offer better performance guarantees than a more centralized division.  Hence, this paper demonstrates that the resource division problem is highly non-trivial in such enmeshed environments and worthy of significant future research efforts. 
\end{abstract}

\section{Introduction}\label{sec:intro}
Multi-agent systems rely on the collective behaviors of independent decision-makers (agents), as they are often too large and complex to allow a centralized authority to effectively operate. Such systems are designed to concurrently accomplish a diverse set of tasks, e.g. multiple organizations contributing to the operation of a supply chain, or a coalition of independent military units sent to secure a number of locations. The agents are often heterogeneous, each possessing distinct roles and/or varying levels of capability. A central problem for a system operator is to determine how to divide available resources among the agents such that they can most effectively accomplish their given tasks. In an ideal setting, each of the tasks can be completed in isolation by a specialized agent, and the optimal division of resources is often straightforward. However, when the completion of a task relies on the behaviors of multiple heterogeneous agents, e.g., the agents' decisions have a degree of interdependence on the completion of the task, the question of how to divide resources may not be as straightforward.


In this paper, we consider such interdependencies in the setting of a Colonel Blotto game, where a team of two sub-colonels compete against a common enemy over the same two battlefields. Each sub-colonel is in control of a portion of the total available resources, and must independently decide how to allocate them across the two battlefields. The sum of the sub-colonels' allocations on each battlefield competes against the enemy's allocation. The measure of system performance we consider here is the optimal security value, which is the highest payoff the team can ensure regardless of the enemy's behavior, through the sub-colonels' independent selection of allocation strategies. 

It is important to understand the limitations of such a distributed decision-making structure in comparison to a completely centralized structure, i.e. where one of the sub-colonels is in control of all the resources. Sub-colonels on a team make decisions independently of each other and, hence, any form of randomization the team can produce as a whole is a result of the players' independent randomizations. This limits the forms of joint randomness a team can produce. A completely centralized structure places no such restrictions on the forms of randomization that can be produced. In this light, performance guarantees for a system with a distributed structure can be no better than a completely centralized structure. Zero-sum games with such team structures have recently been studied \cite{Schulman_2019}. 


The primary focus of the paper is on answering the following question: ``How should $B$ available resources be divided among the two sub-colonels by endowing each with resources $B_1$ and $B_2$ (such that $B_1 + B_2 = B$), in order to maximize their achievable performance guarantees?'' In the extreme case, the choice $B_1 = 0$ and $B_2 = B$ reduces to a completely centralized command structure, where sub-colonel 2 is in control of all $B$ resources. Meanwhile, the case where $B_1, B_2 > 0$ represents a distributed command structure wherein each team sub-colonel has independent control of a portion of the total $B$ resources (see Figure \ref{fig:diagram}). Intuition suggests one should make the system as `centralized' as possible -- as the division $B_1$ increases to $B/2$, we say the system becomes `less centralized', as sub-colonel 2's control of the larger portion approaches sub-colonel 1's portion. Indeed, if the extreme case $B_1 = 0$ is an option, this is a trivial decision to make. However, the centralized option may not always be available to a system operator due to constraints or limitations, e.g. sub-colonel 1 must be in control of a positive portion of the available resources. In the presence of such constraints, is the most centralized option (making one sub-colonel as strong as possible within the constraints) still the best choice to make? 

Our main contribution in this paper, contrary to intuition, asserts that the most centralized option is not the best division of resources in general. In particular, we show that the team's achievable performance guarantees are not, in general, monotonic for $B_1\in[0,B/2]$. Furthermore, we identify non-centralized divisions of the resources in which the team can recover the same performance as the completely centralized case.  Our results suggest that the problem of optimally dividing resources among agents that comprise autonomous systems is highly non-trivial, especially when there are interdependencies between the agents' actions. Hence, an understanding of the particular system at hand is required.

\begin{figure}[t]
    \centering
    \includegraphics[scale = .16]{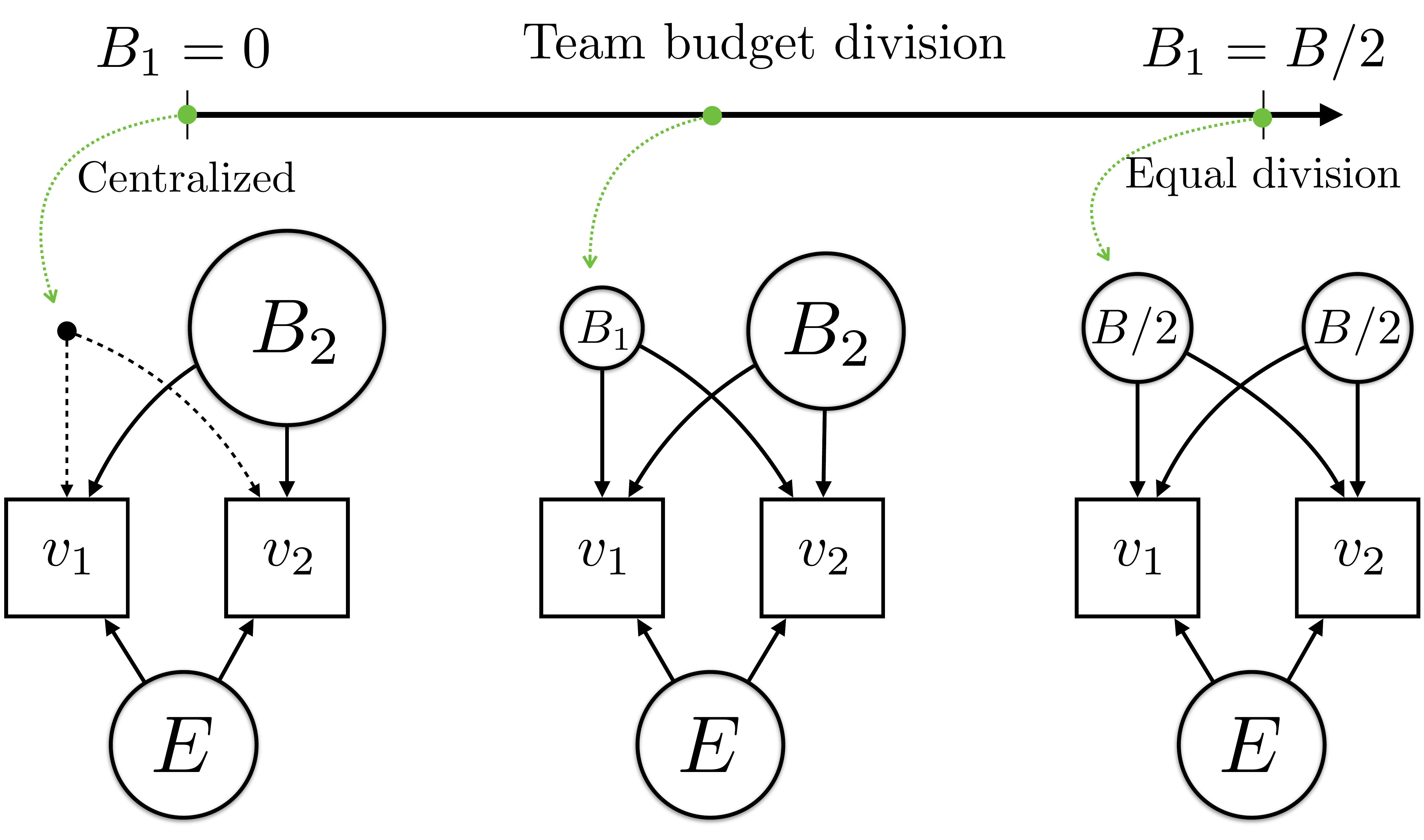}
    \caption{\small The range of command structures in the team Blotto game under consideration in this paper. There is a total of $B$ resources to be divided among sub-colonels 1 and 2, who together compete against the enemy with $E$ resources. Here, $B_1 \in [0,B/2]$ indicates the amount of resources endowed to sub-colonel 1, and hence $B_2 = B- B_1$ resources endowed to sub-colonel 2. The sub-colonels have independent control over their endowed resources, and decide how to allocate them over the same two battlefields. When $B_1 = 0$, the system reduces to a completely centralized command, wherein a single colonel has control over all $B$ resources. This is the classic 1 vs 1 Colonel Blotto game with two battlefields, well studied in \cite{Gross_1950} and \cite{Macdonell_2015}. As $B_1$ increases up to $B/2$, the system becomes less centralized, as the majority portion of resources under player 2's control becomes less dominant.}
    \label{fig:diagram}
\end{figure}

\noindent \textbf{Related works:} Much research in the game theory literature is devoted to characterizing how system performance guarantees can improve through the design of agents' utility functions \cite{Paccagnan_2019,Gairing_2009,Marden_2013}. Optimal designs  facilitate self-interested behaviors that lead to Nash equilibria with good system performance guarantees. Instead of altering agents' utility functions to achieve different system designs, the present paper focuses on how altering the degree of centralization, i.e. through agents' resource endowments, ultimately affects behavior and achievable performance guarantees.


Colonel Blotto games have been studied for 100 years, and are known to be difficult to solve in general. This is largely due to the fact they do not admit pure strategy Nash equilibria \cite{Borel,Golman_2009}. They are commonly formulated as zero- or constant-sum games, and hence equilibrium (mixed) strategies of the opposing colonels are optimal security strategies, i.e. strategies that ensure the highest payoff regardless of the opponent's behavior. The primary literature on Colonel Blotto is concerned with characterizing the value of this highest payoff, or optimal \emph{security value}, in completely centralized settings \cite{Gross_1950,Roberson_2006,Macdonell_2015,Thomas_2018,Kovenock_2020,Vu_2020thesis}. In recent years, simpler variants of Blotto games have been considered to study team settings. For instance, \cite{Kovenock_2012,Gupta_2014b} study coalitional scenarios where two players opposing a common enemy can decide to unilaterally transfer resources among themselves before play begins. The model in \cite{Chandan_2020} considers a similar setup, where a team's players instead decide to pre-commit resources onto battlefields. These models, however, do not incorporate any task interdependence, i.e. the team players compete against the enemy on their own sets of battlefields. In the present paper, we are primarily concerned with the scenario where the players on the team have full overlap over their tasks.

\section{Model}\label{sec:model}

\subsection{Centralized Colonel Blotto game with two battlefields}
Blotto (resp. Enemy) has $B>0$ (resp. $E>0$) resources to allocate over two battlefields. A pure strategy for Blotto (resp. Enemy) is a number $b\in[0,B]$ (resp. $e\in[0,E]$), which is the amount of resources sent to the first battlefield -- the remaining $B-b$ (resp. $E-e$) is thus sent to the second battlefield. Each battlefield $j\in\{1,2\}$ is associated with a value $v_j \geq 0$. Given a strategy profile $(b,e)$, Blotto's payoff is given by
\begin{equation}\label{eq:U}
    U_B(b,e) := v_1 \cdot W(b,e) + v_2 \cdot W(B-b,E-e) 
\end{equation}
where
\begin{equation}\label{eq:W}
    W(x,y) := 
    \begin{cases}
    1, &\text{if } x > y \\
    1/2, &\text{if } x = y \\
    0, &\text{if } x < y
    \end{cases}.
\end{equation}
Enemy's payoff is defined as $U_E(e,b):= v_1 + v_2 - U_B(b,e)$. A mixed strategy for Blotto (resp. Enemy) is any measurable, univariate probability distribution $\Bcdf$ (resp. $\Ecdf$) with compact support $[0,B]$ (resp. $[0,E]$). Here, $\Bcdf$ represents the cumulative distribution function on Blotto's allocation $b$ to the first battlefield. We will use lower case $f$ to denote a distribution's density function. Note that $F_B$ completely determines the probability distribution on $B - b$, the allocation on the second battlefield. The payoff \eqref{eq:U} can be extended to admit mixed strategies, where $U_B(F_B,F_E)$ is the expected payoff with respect to $F_B,F_E$. Let us denote $\Delta(B)$ as the set of all mixed strategies $F_B$ with support on $[0,B]$.

The \emph{value} associated with Blotto's strategy $\Bcdf$ is the worst payoff it attains among Enemy's strategies: 
\begin{equation}
    V(\Bcdf) := \min_{e \in [0,E]} U_B(\Bcdf,e).
\end{equation}
The \emph{security value} is defined as
\begin{equation}
    V^* := \max_{F_B\in\Delta(B)} V(F_B).
\end{equation}
We call a distribution $F_B$ that satisfies $V(F_B) = V^*$ a \emph{security strategy}. Gross and Wagner \cite{Gross_1950} first characterized the security value (equivalently, equilibrium payoff) and some security strategies for the two battlefield Colonel Blotto game. To simplify exposition, we set $v_1 = v_2 = 1$:
\begin{equation}\label{eq:Vstar}
    V^* = 
    \begin{cases}
        1 - \frac{1}{m}, &\text{if } \frac{B}{E} \in (\frac{m-1}{m},\frac{m}{m+1}], \ \text{for } m=1,2,\ldots \\
        1, &\text{if } \frac{B}{E} = 1 \\
        1 + \frac{1}{m}, &\text{if } \frac{B}{E} \in (\frac{m+1}{m}, \frac{m}{m-1}], \ \text{for } m=1,2,\ldots
    \end{cases}
\end{equation}
Hence, $V^*$ is the security value achievable by a completely centralized command structure -- a single player, Blotto, is in control of $B$ resources. We thus refer to $V^*$ as the \emph{centralized security value}. Note the range of budgets is split into a countably infinite number of partitions. We say the budgets are in \emph{partition} $m$ if $\frac{B}{E} \in (\frac{m-1}{m},\frac{m}{m+1}]$.

\subsection{Team Colonel Blotto game with two battlefields}

Blotto's total resource budget $B$ is divided among two sub-players. We will use the terminology `sub-player' or simply `player' instead of `sub-colonel' for the remainder of the paper. Player 1 (resp. player 2) is under control of $B_1$ (resp. $B_2$) resources, with $B_1 + B_2 = B$. Both players have the ability to independently allocate resources to both battlefields -- player $i\in\{1,2\}$ chooses $b_i \in [0,B_i]$ resources to allocate to battlefield 1, and the rest $B_i - b_i$ to battlefield 2. Given $b_1$ and $b_2$, the team's overall resource allocation on battlefield 1 is $b_1 + b_2$, and on battlefield 2 is $B - (b_1 + b_2)$. A mixed strategy for sub-player $i$ is any $F_i \in \Delta(B_i)$. A pair of mixed strategies thus induces $F_B\in \Delta(B)$ on the team's overall allocation on battlefield 1, whose density function is given by the convolution
\begin{equation}
    f_B(x) = (f_1 \circledast f_2)(x) := \int_0^{B_2} f_1(x-s) f_2(s) \,ds
\end{equation}
for all $x \in [0,B]$. With some abuse of notation, we will use $F_B = F_1 \circledast F_2$ to denote the cumulative distribution function of $f_1 \circledast f_2$. Let us define the \emph{distributed security value} as
\begin{equation}\label{eq:distributed_value}
    \begin{aligned}
        V_d^*(B_1) := &\max_{F_B \in \Delta(B)} V(F_B) \\ 
        \quad \text{s.t. } &F_B = F_1 \circledast F_2, \ F_i \in \Delta(B_i), i = 1,2 \\
        &B_2 = B - B_1 
    \end{aligned}
\end{equation}
and a \emph{distributed security strategy} as a pair $F_i \in \Delta(B_i)$, $i=1,2$, that satisfies $V(F_1\circledast F_2) = V_d^*(B_1)$. Since any $F_B=F_1 \circledast F_2$ is a member of $\Delta(B)$, the relation $V_d^* \leq V^*$ follows immediately.

\subsection{Numerical examples and discussion}

Note that setting $B_1 = 0$ recovers the centralized Blotto game, as all $B$ resources are under the control of player 2. Consequently, $V_d^*(0) = V^*$. The choice of how to divide the sub-players' budgets so as to maximize $V_d^*(B_1)$ is thus a trivial task if the centralized option $B_1 = 0$ is available. However, suppose the constraint $B_1 \in [\alpha B, B/2]$ must hold, i.e. player 1 must be in control of at least a positive fraction $\alpha \in (0,1/2)$ of the total resources. Intuition suggests that the division should be made as `centralized' as possible, i.e. setting $B_1 = \alpha B$ leaves $B_2$ with the largest possible portion. Indeed, less centralized divisions restrict the team's ability to employ the jointly randomized strategies necessary for optimizing their security value. Does it hold that 'more centralized' divisions always do better than less centralized divisions? Specifically, is $V_d^*(B_1)$ a monotonically decreasing function in $B_1 \in [0,B/2]$?

To see if this intuition holds, we performed numerical evaluations on an integer version of the team Blotto game, as detailed in the following example. We stress that the following example is provided solely to develop intuition and does not serve as a valid proof for our forthcoming analytical results in Section \ref{sec:proof}.
\begin{figure}
    \centering
    \includegraphics[scale=.5]{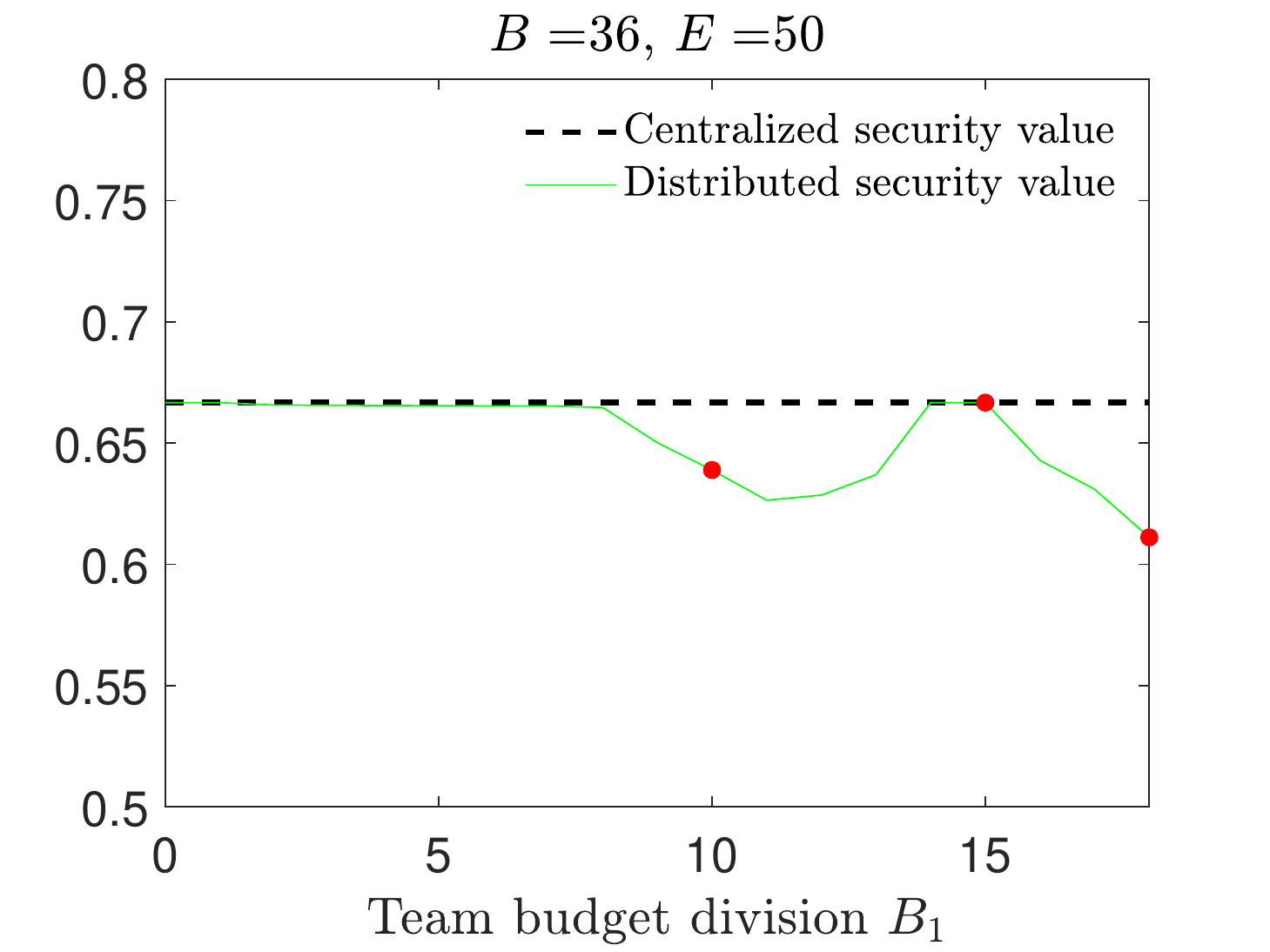}
    \caption{\small Computed security values in an integer team Blotto game using the numerical scheme described in Example 1. Here, $E = 50$, $B=36$, and $B_1 \in \{0,1,\ldots,18\}$. The distributed security value appears to be non-monotonic in the size of the division $B_1$, contrary to intuition. There are divisions, e.g. $B_1 = 15$, $B_2 = 21$, that perform as well as the centralized case. Moreover, this division outperforms more centralized divisions, e.g. $B_1 = 10$, $B_2 = 26$. These points are marked by red dots. }
    \label{fig:slice_36}
\end{figure}
\begin{example}
    Consider an integer Blotto game, i.e. allocations to battlefields are restricted to be integers. The sub-players' mixed strategies are probability vectors of length $B_i+1$, which specify the randomization over all possible allocations $\{0,1,\ldots,B_i\}$ to battlefield 1. We can thus re-formulate \eqref{eq:distributed_value} as a finite-dimensional optimization problem that is non-convex, due to the convolution constraint. We then used numerical optimization techniques to find the distributed security value in this setting. In particular, we applied the nonlinear function solver \texttt{fmincon} in Matlab, to solve the re-formulation of \eqref{eq:distributed_value}. We stress here that the computed security values $V_d^*$ from this scheme may not be completely accurate, as \eqref{eq:distributed_value} is highly non-convex and the nonlinear function solver is not guaranteed to converge to the optimal point. As such, one would treat any resulting numerical computation as a lower bound on the actual security value.  We use such numerical tools here to simply gauge the behavior of $V_d^*$, and to develop our intuition for general theoretical properties one might establish on $V_d^*$ in the non-integer setting.

    Now, consider an enemy budget of $E=50$, and the total resources $B = 36$ are divided among the two sub-players in the range $B_1 \in \{0,\ldots,18\}$. Figure \ref{fig:slice_36} depicts the computed distributed security values in this range. Most notably, $V_d^*$ does not appear to be monotonic in the division $B_1$. Moreover, there are less centralized divisions (e.g. $B_1 = 15$, $B_2 = 21$) that provide better performance guarantees than more centralized divisions (e.g. $B_1 = 10$, $B_2 = 26$).
\end{example}
Our numerical study suggests that our intuition with respect to more centralized divisions of resources always performing better than less centralized divisions is incorrect. While this is merely a numerical study on a single instance of an approximate, integer version of the class of games we consider, it raises interesting questions about how resources should be distributed among multiple team members. In particular, we seek to establish analytically whether $V_d^*$ is in fact, non-monotonic. Moreover, Figure \ref{fig:slice_36} also suggests there are less centralized divisions that can recover the completely centralized security value, whereas slightly more centralized divisions cannot.

In the next section, we identify a broad class of instances of the (non-integer) team Blotto game where such properties do in fact hold. In particular, $V_d^*$ is not monotonic in general, and one can find disjoint intervals within $B_1 \in [0,B/2]$ that correspond with divisions that recover the centralized security value. These properties demonstrate that the resource division problem is highly non-trivial in such interdependent multi-agent environments.

\section{Main results}\label{sec:proof}

In this section, we focus on the (non-integer) team Blotto game and identify a number of non-intuitive properties of the distributed security value $V_d^*$. In particular, we establish for a broad class of instances that there exist disjoint intervals within $B_1 \in [0,B/2]$ corresponding with divisions where the centralized security value can be recovered from a distributed command structure (Proposition \ref{prop:bands}). Additionally, and most importantly, we establish the following:
\begin{theorem}\label{thm:nonmonotone}
    The distributed security value $V_d^*(B_1)$ is not, in general, a monotonic function of $B_1 \in [0,B/2]$.
\end{theorem}
\begin{remark}
    The statement of Theorem 1 would hold even if it is true for only a single game instance. However, our approach to verify Theorem 1 studies a wide range of game instances where we are able to prove the non-monotonicity of $V_d^*(B_1)$. The instances we identify do not exhaust all two battlefield Blotto games. Indeed, there may be an even broader range of instances for which non-monotonicity holds (left for future work). Nonetheless, our analysis demonstrates that the non-monotonicity property is not an anomalous edge case.
\end{remark}

Our approach to proving Proposition \ref{prop:bands} and Theorem 1 is as follows. We first state the necessary and sufficient conditions on $F_B$ to be a (centralized) security strategy. Call $\Omega_B$ the set of all centralized security strategies. We then show on particular disjoint intervals of divisions within $B_1 \in [0,B/2]$, one can reconstruct a security strategy $F_1 \circledast F_2 \in \Omega_B$ (Proposition \ref{prop:bands}).  We then identify a class of game instances parameterized by the budgets $B/E$ for which there are at least two such intervals, and characterize a range of divisions $B_1$ that lie between two intervals where $F_B = F_1 \circledast F_2 \notin \Omega_B$ for any $F_1\in\Delta(B_1)$, $F_2\in\Delta(B_2)$, and hence $V_d^*(B_1) < V^*$. This fact establishes Theorem 1.

Throughout, we will assume that $B < E$ and $v_1 = v_2 = 1$ to simplify exposition. The arguments can be generalized to $v_1 \neq v_2$. Let $d:= E-B$ be the budget difference. In partition $m$, i.e. $B/E \in (\frac{m-1}{m},\frac{m}{m+1}]$, it holds that $(m-1)d < B \leq md$. We define $r_B := B - (m-1)d \in (0,d]$. We denote $[m]$ as the set of integers $\{1,\ldots,m\}$.
\begin{lemma}[Necessary and sufficient conditions for centralized security strategies]\label{lem:SS_conditions}
    Suppose $B/E<1$ and $r_B < d$. Then $F_B \in \Omega_B$ if and only if
    \begin{align}
		\int_{I_j} dF_B &= 1/m, \quad \forall j \in [m] \label{eq:SS_property1} \tag{SS-1}  \\
		\int_{I_{j+1} \cap [0,x]} dF_B &\leq \int_{I_j \cap [0,x-d)} dF_B \tag{SS-2} \label{eq:SS_property2}  \\ 
		& \quad \forall j \in [m-1], \ \forall x \in I_{j+1} \nonumber
	\end{align}
	where $I_1 := [0, r_B]$ and $I_j := ((j-1)d, (j-1)d + r_B]$ for all $j=2,\ldots,m$.
\end{lemma}
Intuitively, \eqref{eq:SS_property1} says a security strategy must have equal probability mass located in small intervals spaced $d$ apart. Condition \eqref{eq:SS_property2} states the probability mass must be placed in such a way that prevents the Enemy from having an allocation $e \in [0,E]$ such that $U_B(F_B,e) < V^*$.
\begin{remark}
    Conditions \eqref{eq:SS_property1} and \eqref{eq:SS_property2} are special cases of the properties identified in \cite{Macdonell_2015} that ensure equilibrium exchangeability in a more general class of two battlefield Blotto games. That is, any $F_B$ satisfying these properties, paired with any $F_E$ satisfying similar properties, forms a Nash equilibrium. While these properties satisfy sufficiency -- any equilibrium strategy in a zero or constant-sum game is also a security strategy -- our proof of Lemma \ref{lem:SS_conditions} also establishes necessity.  
\end{remark}
\begin{proof}
    Enemy's payoff from using a pure strategy $x \in [0,E]$ against the strategy $F_B \in \Delta(B)$ can be expressed as
    \begin{equation}
        U_E(x,F_B) = F_B(x) - F_B(x-d) + 1.
    \end{equation}
    Recall enemy's security value is given by $1+1/m$ in partition $m$ \eqref{eq:Vstar}.
    
	\noindent $\Leftarrow:$  Suppose properties \eqref{eq:SS_property1} and \eqref{eq:SS_property2} hold. Define $D(x) := F_B(x) - F_B(x-d)$ for $x \in [0,E]$. It suffices to show that $\max_{x\in[d,B]} D(x) = 1/m$. Indeed, $D(x) = 1/m$ for $x \in \{jd\}_{j=1}^{m-1} \cup \{jd+r_B\}_{j=1}^{m-1}$, by property \eqref{eq:SS_property1}. Furthermore, $D(x) \leq 1/m$ for any $x \in (jd, jd+r_B)$ and any $j \in [m-1]$, by property \eqref{eq:SS_property2}.
	
	\noindent $\Rightarrow$: Suppose $V(F_B) = V^*$, i.e. it holds that
	\begin{equation}\label{eq:max_prob}
		\max_{x \in [d,B]} D(x) = 1/m.
	\end{equation}
	Suppose $F_B$ satisfies property \eqref{eq:SS_property1} but not property \eqref{eq:SS_property2}. Then there exists a $k\in[m-1]$ and $x \in I_{k+1}$ such that $\int_{I_{k+1} \cap [0,x]} dF_B > \int_{I_k \cap [0,x-d)} dF_B$. Hence,
	\begin{equation}
		\begin{aligned}
			D(x) = \int_{I_{k+1} \cap [0,x]} dF_B + \frac{1}{m} - \int_{I_k \cap [0,x-d)} dF_B > 1/m
		\end{aligned}
	\end{equation}
	which contradicts \eqref{eq:max_prob}. Now, suppose $F_B$ does not satisfy property \eqref{eq:SS_property1}. Let $S = [0,B]\setminus \bigcup_{j=1}^m I_j$. We split into two scenarios. First, suppose $\int_S dF_B = 0$. Then there is a $k\in[m]$ with $\int_{I_k} dF_B > 1/m$, contradicting \eqref{eq:max_prob}. Second, suppose $\int_S dF_B > 0$. Define $m$ collections of intervals $Y^i := \{Y_j^i\}_{j\in[m]}$ as follows: for each $i\in[m]$,
	\begin{equation}
		\begin{aligned}
            Y_i^i &= ((i-1)d, (i-1)d + r_B] = I_i &  \\
			Y_1^i &= [0,d] &(\text{if } i \neq 1) \\
			Y_j^i &= ((j-1)d,jd], &j=2,\ldots,i-1 \\
			Y_j^i &= ((j-2)d + r_B, (j-1)d + r_B], &j=i+1,\ldots,m
		\end{aligned}
	\end{equation}
	By construction, $Y_j^i \cap Y_k^i = \varnothing$ for any $j,k \in [m]$ and $\bigcup_{j\in[m]} Y_j^i = [0,B]$. Note the length of each $Y_j^i$, $j \neq i$, is precisely $d$, and the length of $Y_i^i$ is $0 < r_B < d$. By \eqref{eq:max_prob}, it must hold that $\int_{Y_j^i} dF_B = 1/m$ for every $i,j \in [m]$. Consequently, it must also hold that $\int_{I_j} dF_B = 1/m$ for all $j \in [m]$. We  obtain $\int_S dF_B = 0$, a contradiction. This establishes the result. Note in the latter scenario we do not make any assumption on whether property \eqref{eq:SS_property2} is satisfied or not.
\end{proof}

Our next result identifies the disjoint intervals of divisions within $B_1 \in [0,B/2]$ for which the distributed security value coincides with the centralized security value.
\begin{proposition}\label{prop:bands}
	If 
	\begin{equation}
    	\begin{aligned}
    	    B_1 &\in \bar{I}_{k_1},
    	\end{aligned}
	\end{equation}
	where $k_1$ is any factor of $m$, then $V_d^*(B_1) = V^*$. Here, $\bar{I}$ indicates the closure of an interval $I$.
\end{proposition}

\begin{proof}
	The strategy $F_B \in \Delta(B)$, whose density is given by
	\begin{equation}\label{eq:d_sep}
	    f_B(x) = \frac{1}{m} \sum_{j=0}^{m-1} \delta(x-jd)
	\end{equation}
	satisfies \eqref{eq:SS_property1} and \eqref{eq:SS_property2}, and thus is a member of $\Omega_B$. Now, let $k_1$ be any factor of $m$. The approach is to reconstruct $F_B$ \eqref{eq:d_sep} through the convolution $F_1 \circledast F_2$. Consider $f_1(x) = \frac{1}{k_1}\sum_{j=0}^{k_1-1} \delta(x-jd)$ and $f_2(x) = \frac{1}{k_2}\sum_{j=0}^{k_2-1} \delta(x-j\cdot k_1d)$, where $k_2$ is such that $k_1 k_2 = m$. Then $(f_1 \circledast f_2)(x) = f_B(x)$ for all $x \in [0,B]$ -- we have reconstructed $F_B$ through the convolution of two independent strategies. However, $F_1$ and $F_2$ must also be feasible for the budget division $B_1$, $B_2$. They are feasible for the range of budgets  $B_1 \in \left[(k_1 - 1)d, (k_1 - 1)d + (B - (m-1)d) \right]$.
\end{proof}
Because 1 is a factor of $m$, the centralized security value can always be achieved on the ``edge" interval $B_1 \in I_1 = [0, r_B]$. To prove the non-monotonicity of $V_d^*$, we need to show  $V_d^*(B_1) < V^*$ for a division $B_1$ that lies in between consecutive intervals $I_{k}, I_{k'}$, where $k,k'$ are consecutive factors of $m$. In particular, let  us focus on when $m > 2$ is even: $I_1$ and $I_2$ are two such consecutive intervals, and any $B_1 \in (r_B,d)$ lies in between them. In the next result, we identify necessary and sufficient conditions for which a distribution $F_B$ satisfies the first property \eqref{eq:SS_property1} for a centralized security strategy. We first need the following definitions.
\begin{definition}
	Given an interval $P = [p^\ell,p^r] \subseteq [0,X]$ and $F \in \Delta(X)$ such that $\int_P dF > 0$, we say $P$ is \emph{reduced} with respect to $F$ if
	\begin{equation}
		P = \bigcap \left\{P' = [p'^\ell,p'^r] \subseteq P : \int_{P'} dF = \int_P dF \right\}
	\end{equation}
\end{definition}
Consequently, for a reduced interval $P$ with respect to $F$, it holds that $\int_{[p^\ell,p^\ell+\epsilon)} dF$ and $\int_{(p^r -\epsilon,p^r]} dF > 0$ for any $\epsilon > 0$.

\begin{lemma}\label{lem:SS1_conditions}
    Suppose $B/E < 1$, $m$ is even, and $B_1 \in (d-r_B,d)$. Then $F_B = F_1 \circledast F_2$ satisfies \eqref{eq:SS_property1}, with $F_1 \in \Delta(B_1)$ ($F_2 \in \Delta(B_2)$), if and only if there exist reduced intervals $\{P_i = [p_i^\ell,p_i^r] \}_{i=1,2}$ with respect to $F_1$ and $\{Q_i = [q_i^\ell,q_i^r]\}_{i=1}^{m/2}$ with respect to $F_2$ such that for all $i \in [m/2]$,
    \begin{align}
        [p_1^\ell+q_i^\ell,p_1^r+q_i^r] &\subseteq I_{2i-1} \text{ and } [p_2^\ell+q_i^\ell,p_2^r+q_i^r] \subseteq I_{2i} \label{eq:even_spacing}
	\end{align}
	and it holds that
	\begin{align}
	    \int_{P_1} dF_1 &= \int_{P_2} dF_1 = 1/2 \label{eq:even_mass_location1} \\
		\int_{Q_i} dF_2 &= 2/m, \ \forall i \in [m/2] \label{eq:even_mass_location2}
	\end{align}
\end{lemma}

\begin{proof}
	Note that any $\{P_i\}_{i=1}^{2}$, $\{Q_i\}_{i=1}^{m/2}$  satisfying \eqref{eq:even_spacing} implies that $(p_2^\ell - p_1^r) - (q_i^r - q_i^\ell) \geq d - r_B$. Since $p_2^\ell - p_1^r \leq B_1$ and $q_i^r - q_i^\ell \geq 0$, this requires that $B_1 \geq d - r_B$. Furthermore, $B_1 \in (r_B,d)$ implies $B_2 \in ((m-2)d+r_B, (m-1)d)$.
	
	\noindent $\Leftarrow$: Conditions \eqref{eq:even_spacing} - \eqref{eq:even_mass_location2} are clearly sufficient conditions for $F_B = F_1\circledast F_2$ to satisfy \eqref{eq:SS_property1}. Intuitively, $F_1$ concentrates equal probability mass at the ends of the interval $[0,X_1]$, and $F_2$  places equal probability mass at $m/2$ locations spaced roughly $2d$ apart.
	When $F_1$ and $F_2$ are convolved, $F_1$ is `duplicated' $m/2$ times at the locations $\{q_i^\ell\}_{i=1}^{m/2}$, and the resulting mass is contained in $\bigcup_{j\in[m]} I_j$.
	
	\noindent $\Rightarrow$: Suppose $F_B = F_1\circledast F_2$ satisfies condition \eqref{eq:SS_property1}, i.e. $\int_{I_j} dF_B = 1/m$ for all $j\in[m]$. Denote $S = [0,B]\setminus \bigcup_{j\in[m]} I_j$. If there do not exist any $\{P_i\}_{i=1}^{2}$, $\{Q_i\}_{i=1}^{m/2}$ that satisfy \eqref{eq:even_spacing}, then for some $k \in [m]$, one cannot find any reduced intervals $P$ (w.r.t $F_1$) and $Q$ (w.r.t $F_2$) such that $[p^\ell + q^\ell, p^r+q^r] \subseteq I_k$. It must hold that $\int_{I_k} F_B = 0$, contradicting \eqref{eq:SS_property1}
	
	Now, let us assume reduced intervals $\{P_i\}_{i=1}^{2}$, $\{Q_i\}_{i=1}^{m/2}$ that satisfy \eqref{eq:even_spacing} do exist.
	Suppose (for sake of contradiction) any such reduced intervals $\{P_i\}_{i=1}^{2}$, $\{Q_i\}_{i=1}^{m/2}$ do not satisfy \eqref{eq:even_mass_location1} and \eqref{eq:even_mass_location2}. Furthermore, suppose that 
	\begin{itemize}
		\item $Q_i \supseteq Q_i'$ for any other reduced $Q_i'$ such that $P_1,P_2,\{Q_i',Q_{-i}\}$ satisfies \eqref{eq:even_spacing}.
		\item $P_1 \supseteq P_1'$ for any other reduced $P_1'$ such that $P_1',P_2,\{Q_i\}_{i=1}^m$ satisfies \eqref{eq:even_spacing}.
		\item $P_2 \supseteq P_2'$ for any other reduced $P_2'$ such that $P_1,P_2',\{Q_i\}_{i=1}^m$ satisfies \eqref{eq:even_spacing}. 
	\end{itemize}
	Intuitively, $\{P_i\}_{i=1}^{2}, \{Q_i\}_{i=1}^{m/2}$ form a `largest' set of reduced intervals that still satisfy \eqref{eq:even_spacing}. It holds that
	\begin{enumerate}[label=(\alph*)]
		\item $\int_{P_1\cup P_2} dF_1 = 1$ and $\int_{P_1} dF_1 \neq 1/2$, or
		\item $\int_{P_1\cup P_2} dF_1 < 1$, or
		\item $\sum_{i=1}^{m/2} \int_{Q_i} dF_2 < 1$, or
		\item $\sum_{i=1}^{m/2} \int_{Q_i} dF_2 = 1$ and $\int_{Q_k} dF_2 \neq 2/m$ for some $k$.
	\end{enumerate}
	Here, (a) and (b) are mutually exclusive, as are (c) and (d). We proceed by showing (a) and (c), (a) and (d), and then (b) holding regardless of whether (c) or (d) holds, leads to a contradiction of \eqref{eq:SS_property1}.
	
	Suppose (a) is true. If \eqref{eq:even_mass_location2} holds (but not \eqref{eq:even_mass_location1}), then  $\int_{I_j} dF_B \neq 1/m$ for any $j\in[m]$, contradicting condition \eqref{eq:SS_property1}. Now, assume \eqref{eq:even_mass_location2} is not true. If (c) holds,
	then there exists a reduced interval $T = [t^\ell,t^r]$ w.r.t $F_2$ disjoint from the $\{Q_i\}_{i=1}^{m/2}$ such that $q_k^r < t^\ell \leq t^r < q_{k+1}^\ell$  for some $k \in \{0,\ldots,m/2\}$ (defining $q_0^r = 0$ and $q_{m/2+1}^\ell = X_2$) and $\int_{T} dF_2 > 0$. It must be the case that $\int_S dF_B > 0$ because  the $\{Q_i\}_{i\in[m/2]}$ are already a `largest set' of reduced intervals. If this was not the case, $Q_k$ and $Q_{k+1}$ could be re-defined to include the probability mass contained in $T$ and still satisfy \eqref{eq:even_spacing}. This leads to a contradiction of \eqref{eq:SS_property1}. If (d) holds, then all mass in contained within the $\{I_j\}_{j\in[m]}$, but $\int_{I_j} dF_B \neq 1/m$ for at least one $j$.

	Suppose (b) is true. Then one can find a reduced interval $T = [t^\ell,t^r]$ w.r.t $F_1$ where $t^\ell > p_{i-1}^r$ and $t^r < p_i^\ell$ for some $i\in\{1,2,3\}$ (defining $p_0^r = 0$ and $p_3^\ell = B_1$), that satisfies $\int_{T} dF_1 > 0$. By the same arguments as above, it must be that $\int_S dF_B > 0$. If this was not the case, $P_1$, $P_2$, or both could be re-defined to include probability mass contained in $T$ and still satisfy \eqref{eq:even_spacing}. This leads to a contradiction of \eqref{eq:SS_property1}. Note this assertion is made irrespective of whether \eqref{eq:even_mass_location2}, (c), or (d) holds or not.
\end{proof}

The final lemma we will need to establish Theorem 1 asserts that no $F_B = F_1\circledast F_2 \in \Omega_B$ in the range $B_1 \in (d-r_B,d)$ can give $V_d^*(B_1) = V^*$.
\begin{lemma}\label{lem:no_SS2}
     Suppose $B/E < 1$, $m$ is even, $B_1 \in (d-r_B,d)$, and $F_B = F_1 \circledast F_2$ satisfies \eqref{eq:SS_property1}. Then $F_B \notin \Omega_B$.
\end{lemma}
\begin{proof}
	By the previous lemma, there exist reduced intervals $P_1,P_2,\{Q_i\}_{i=1}^{m/2}$ that satisfy \eqref{eq:even_spacing} - \eqref{eq:even_mass_location2}, such that $F_B$ satisfies \eqref{eq:SS_property1}: $\int_{I_j} dF_B = 1/m$ for $j\in[m]$. Here,  
	\begin{equation}
		\text{supp}(F_B) = \bigcup_{i=1}^{m/2} [p_1^\ell + q_i^\ell , p_1^r + q_i^r] \cup [p_2^\ell + q_i^\ell, p_2^r + q_i^r]
	\end{equation}
	where $[p_1^\ell + q_i^\ell, p_1^r + q_i^r] \subset I_{2i-1}$ and $[p_2^\ell + q_i^\ell, p_2^r + q_i^r] \subset I_{2i}$ for each $i \in [m/2]$. These intervals are  reduced w.r.t $F_B$, since they were generated from a convolution of reduced intervals. Because $B_1 < d$, we have $p_2^r+q_i^r - (p_1^r+q_i^r) = p_2^r - p_1^r < d$. Consequently, the interval $(p_2^r+q_i^r-d,p_2^r+q_i^r]$ contains the $1/m$ mass in $I_{2i}$ in addition to a nonzero mass in $I_{2i-1}$:
	\begin{equation}
	    \int_{(p_2^r+q_i^r-d,p_2^r+q_i^r]} dF_B > 1/m.
	\end{equation}
	Hence, \eqref{eq:SS_property2} is not satisfied.
\end{proof}
Lemmas \ref{lem:SS_conditions} - \ref{lem:no_SS2} and Proposition \ref{prop:bands} thus establish Theorem 1. Proposition \ref{prop:bands} asserts any game in an even partition $m > 2$ (implying $B/E > 2/3$) will have two intervals, $[0,r_B]$ and $(d,d+r_B]$, within $B_1 \in [0,B/2]$ where $V_d^*(B_1) = V^*$. Lemmas \ref{lem:SS1_conditions} and \ref{lem:no_SS2} show there are divisions $B_1$ between these two intervals such that one can find strategies $F_B = F_1\circledast F_2$ that satisfy \eqref{eq:SS_property1}, but never \eqref{eq:SS_property2}. Hence, $V_d^*(B_1) < V^*$ for these divisions. Note that while showing non-monotonicity of the distributed security value for only a single instance of $B/E$ is required to prove the statement of Theorem 1, we have done so for a broad class of instances.

\begin{figure*}
    \centering
    \includegraphics[scale=.48]{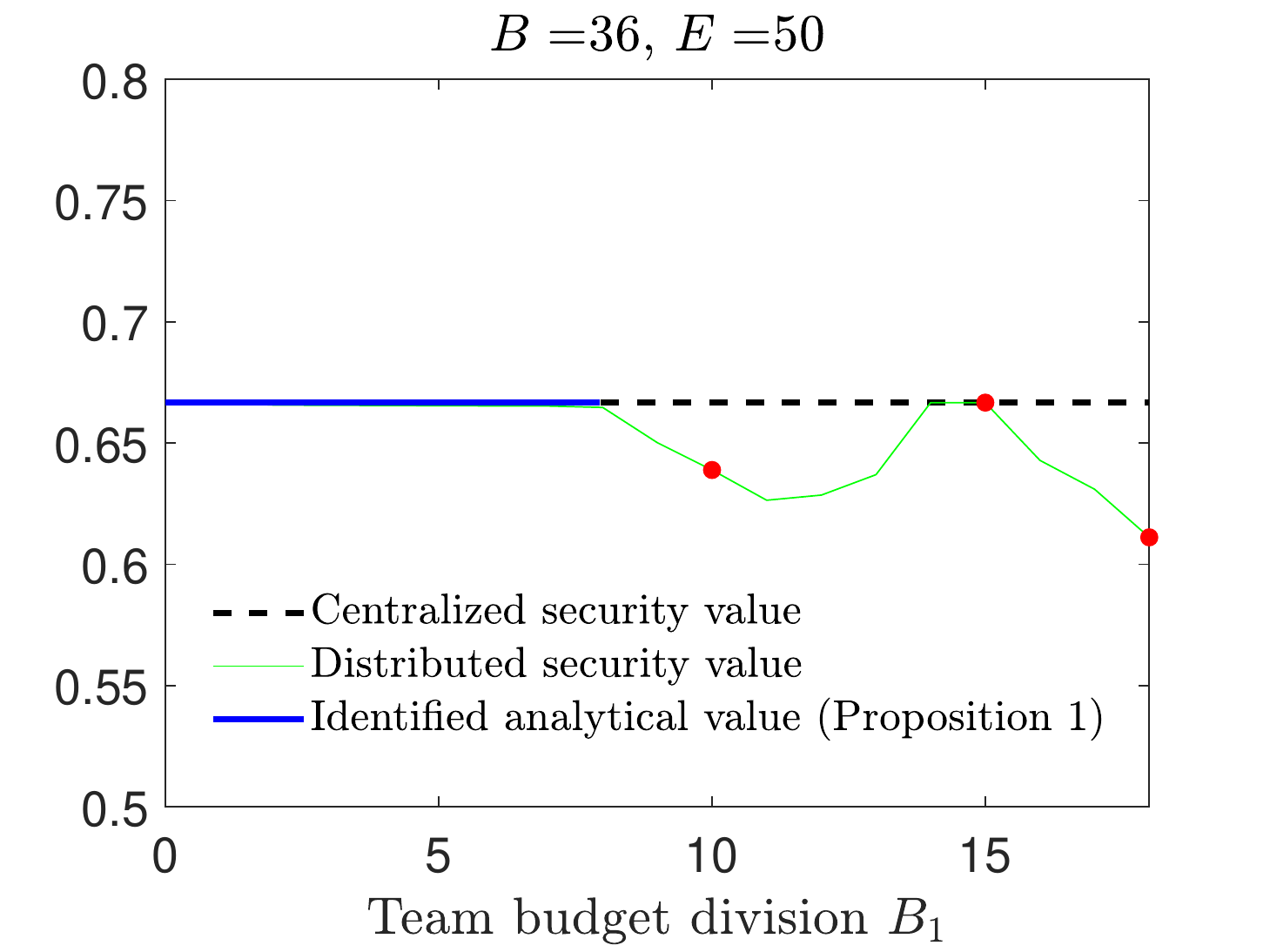}
    \includegraphics[scale=.48]{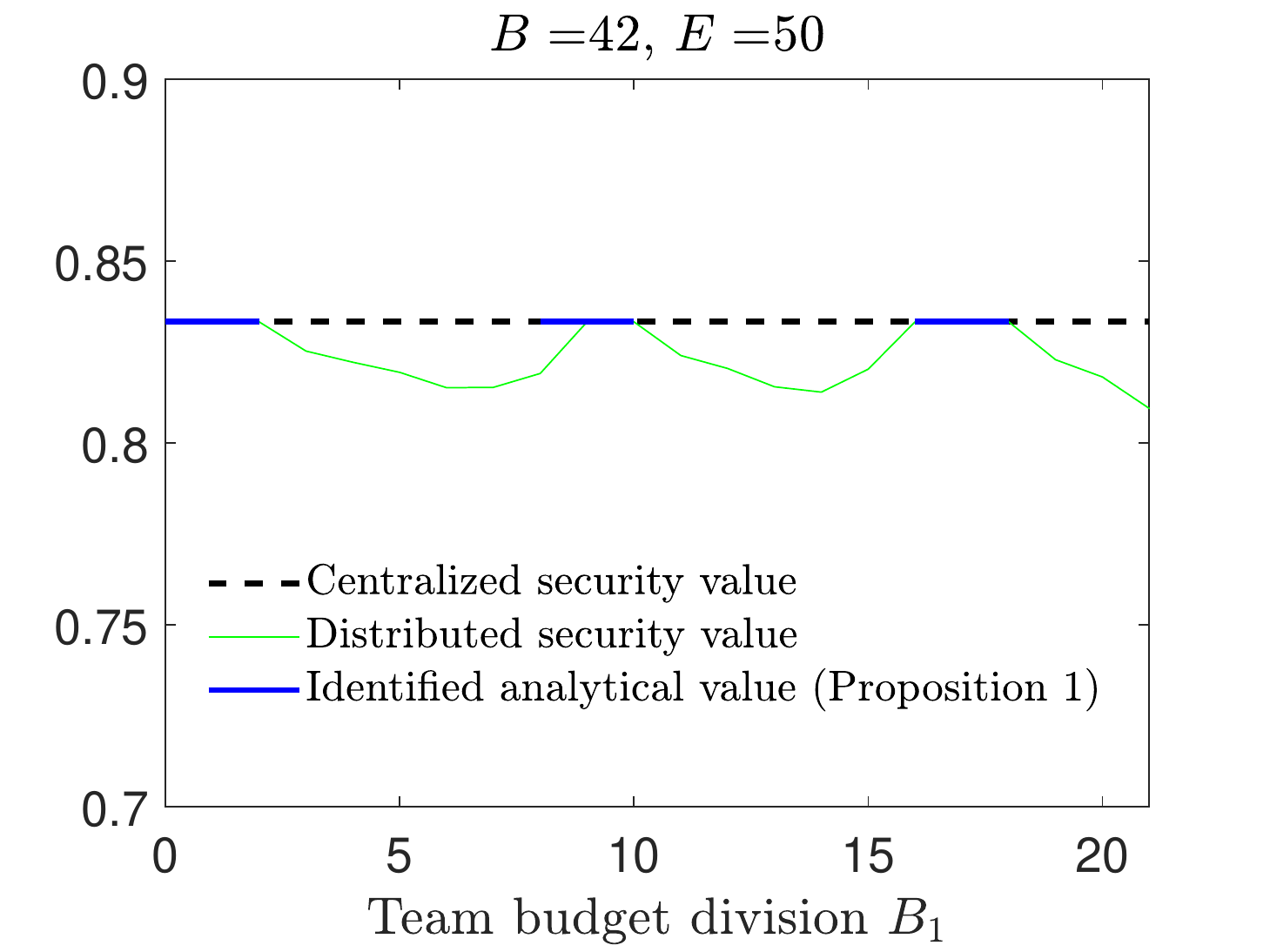}
    \caption{\small Plots of computed distributed security values in the integer version of the team Blotto game. They display non-monotonic behavior as the division $B_1$ ranges from 0 (completely centralized) to 18 (distributed, most balanced budgets). (Left) $B=36$, $E=50$. The red dots indicate the particular setups studied in Example 1. This game instance is in partition $m=3$. Proposition \ref{prop:bands} states the distributed security value coincides with the centralized (for non-integer version) in the interval $B_1 \in [0,r_B]$ (blue line), where $r_B=8$ here. The simulation is in accordance with this prediction. As Proposition \ref{prop:bands} is only a sufficient condition for $V_d^*(B_1) = V^*$, it does not account for the distributed value at $B_1 = 15$ shown here that recovers the centralized security value. Further study is required to verify and to find distributed strategies that recover the centralized value here. (Right) $B=42$, $E=50$. This game instance is in an even partition ($m=6$), hence the theoretical arguments for non-monotonicity in Section \ref{sec:proof} apply. In particular $V_d^*(B_1) < V^*$ in the identified interval $B_1\in (r_B,d)$, where $r_B = 2$ and $d=8$. Proposition \ref{prop:bands} states $V_d^*(B_1) = V^*$ on the intervals $B_1 \in [0,r_B]$, $(d,d+r_B]$, and $(2d,2d+r_B]$. The simulation is in accordance with these predictions, with the exception of the division $B_1 = 8$. Indeed, \texttt{fmincon} does not guarantee the global solution to a nonlinear, non-convex optimization. Any such solution it returns would serve as a lower bound.}
    \label{fig:slices}
\end{figure*}

\section{Simulations}

In this section, we provide some numerical simulations that further highlight the non-monotonic nature of the distributed security value $V_d^*$. As described in Example 1 in Section \ref{sec:model}, we employ numerical techniques to calculate $V_d^*(B_1)$ \eqref{eq:distributed_value} in the analogous integer Blotto game, where allocations to battlefields are restricted to the integers. The reformulated optimization problem \eqref{eq:distributed_value} is thus finite-dimensional, but remains non-convex. Here, the sub-players' mixed strategies are probability vectors of length $B_i+1$ that specify their independent randomizations over all possible allocations to battlefield 1. In particular, their strategy spaces are $\{(0,B_i),(1,B_i-1),\ldots,(B_i,0)\}$. The optimization is non-convex because of the convolution constraint -- the joint probability vector over the strategy space $\{(0,B),\ldots,(B,0)\}$ must be a product distribution from the sub-players' mixed strategies. We implemented the nonlinear function solver \texttt{fmincon} in Matlab to solve this optimization problem.

Once again, we stress that the computed security values $V_d^*$ from this scheme may not be completely accurate, as \eqref{eq:distributed_value} is non-convex and the nonlinear function solver is not guaranteed to converge to the optimal point. As such, one should treat any resulting numerical computation as a lower bound on the actual security value.  We used such numerical tools in Section \ref{sec:model} (Figure \ref{fig:slice_36}) to gauge the behavior of $V_d^*$ and to develop intuition for general theoretical properties on $V_d^*$. We proceeded to establishing such properties analytically in Section \ref{sec:proof} for the non-integer Blotto setting.


Figure \ref{fig:slices} depicts resulting distributed security values computed from our numerical scheme, over the range of budget divisions $B_1 \in \{0,1,\ldots,B/2\}$. These plots illustrate the non-monotonic behavior of $V_d^*(B_1)$ and validate Proposition \ref{prop:bands}, which identifies intervals within $[0,B/2]$ where $V_d^*(B_1) = V^*$.

\section{Conclusion}

Multi-agent systems allow many agents to autonomously collaborate on accomplishing complex tasks by distributing decision-making abilities and shared resources among them. However, when the actions of multiple agents are interdependent with regards to completing the same task, inefficiencies can arise as a result of their independent decision-making processes. In this paper, we framed such a scenario in the context of a Colonel Blotto game, where a team of two players compete against a common enemy over the same two battlefields. We studied how the division of resources among the two players affects their chances against the enemy. The divisions range from a completely centralized command structure, i.e. one of the team players has control over all resources, to varying degrees of distributed command structures, where each player has control over a portion of the total resources. Our main contribution asserts that the team's performance is non-monotonic in this range of command structures. This finding alludes to an interesting design problem in distributing resources among autonomous agents who are collaborating together on a complex task.

While we have established the interesting role of resource division on the team's performance in a symmetric team setting, this paper is clearly a first step in studying a range of research questions on this topic. For instance, one would like to characterize the set of Nash equilibria in these team settings, and identify any inefficiencies that can arise in such stable outcomes. One can then consider utility design problems as a means to coordinate the team's behavior. The presence of multiple concurrent tasks and interdependencies for the team generalizes our current setting, and is also worthy of study. 



\bibliographystyle{IEEEtran}
\bibliography{sources}

\end{document}